\documentclass[preprint,12pt]{elsarticle}

\usepackage{amssymb}
\usepackage{amsmath}
\usepackage{longtable}
\usepackage{amsthm}

\newtheorem{theorem}{Theorem}[section]
\newtheorem{lemma}[theorem]{Lemma}

\newtheorem{corollary}[theorem]{Corollary}
\theoremstyle{definition}
\newtheorem{definition}[theorem]{Definition}
\theoremstyle{remark}

\begin{document}

\begin{frontmatter}



\title{Effective Construction of a Class of Bent Quadratic Boolean  Functions}

\author[label1]{Chunming Tang}
\author[label2,label3]{Yanfeng Qi}

\address[label1]{School of Mathematics and Information, China West Normal University, Sichuan Nanchong, 637002, China}
\address[label2]{LMAM, School of Mathematical Sciences, Peking University, Beijing, 100871, China}
\address[label3]{Aisino Corporation Inc., Beijing, 100195,  China}

\begin{abstract}
In this paper, we consider the characterization of
the bentness of quadratic Boolean functions of the form $f(x)=\sum_{i=1}^{\frac{m}{2}-1}Tr^n_1(c_ix^{1+2^{ei}})+
Tr_1^{n/2}(c_{m/2}x^{1+2^{n/2}}) ,$
where $n=me$, $m$ is even and $c_i\in GF(2^e)$.
For a general $m$, it is difficult to determine the bentness of these functions. We present the bentness of quadratic Boolean function for
two cases: $m=2^vp^r$ and $m=2^vpq$, where $p$ and $q$ are two distinct
primes. Further, we give  the enumeration of quadratic bent functions for the case $m=2^vpq$.
\end{abstract}


\end{frontmatter}


\section{Introduction}
\label{intro}

A bent function, whose Hamming distance to the set of all affine Boolean  functions equals
$2^{n-1}\pm 2^{\frac{n}{2}-1}$, is a Boolean function  with even $n$ variables from $GF(2^n)$ to $GF(2)$. Further, it has maximum  nonlinearity
and the absolute value of its Walsh transform has a constant magnitude \cite{R}. Nonlinearity is an important property for a Boolean function in cryptographic  applications. Much research has been paid on bent functions \cite{CC,C,CCZ,CPT,DLA,HF,KGS1,MLZ,YG}. Since bent functions with maximal nonlinearity have a close relationship with sequences, bent functions are often used in the construction of sequences with maximally linear complexity and low correlation\cite{BK,G,GG,KN, LC,OS, U}. Further, many applications of bent functions can be found in coding theory \cite{MS} and combinatorial design.

As another class of Boolean functions, semi-bent functions are also highly nonlinear.
For an even integer $n$, the Walsh spectra of  bent functions with $n$ variables has the value $\pm 2^{\frac{n}{2}}$ while the Walsh spectra of semi-bent functions
belongs to  $\{0,\pm 2^{\frac{n+2}{2}}\}$. For an odd integer
$n$, the Walsh spectra of semi-bent functions belongs to  $\{0,\pm 2^{\frac{n+1}{2}}\}$. Khoo, Gong and Stinson \cite{KGS,KGS1} considered the quadratic Boolean function of the form
$$
f(x)=\sum_{i=1}^{\frac{n-1}{2}}c_iTr_1^n(x^{1+2^i}),
$$
where $n$ is odd,  $Tr_1^{n}(x)$ is the trace function from $GF(2^n)$ to $GF(2)$ and $c_i\in GF(2)$. They proved that $f(x)$ is semi-bent if and only if
$$
gcd(c(x),x^n+1)=x+1,
$$
where $c(x)=\sum_{i=1}^{\frac{n-1}{2}}c_i(x^i+x^{n-i})$.

Charpin, Pasalic and Tavernier \cite{CPT} generalized Khoo et al.'s results to even
$n$ and considered quadratic functions of the form
$$
f(x)=\sum_{i=1}^{\lfloor\frac{n-1}{2}\rfloor}c_iTr_1^n(x^{1+2^i}),c_i\in GF(2).
$$
When $n$ is even, they proved that $f(x)$ is semi-bent if and only if
$$
gcd(c(x),x^n+1)=x^2+1,
$$
where $c(x)=\sum_{i= 1}^{\frac{n-2}{2}}c_i(x^i+x^{n-i})$.
For odd $n$, they investigated the conditions for the semi-bent functions of
$f(x)$ with three and four trace terms.

For further generalization, Ma, Lee and Zhang \cite{MLZ} applied techniques from \cite{KGS1} and considered the quadratic Boolean functions of the form
\begin{align}\label{mbent}
f(x)=\sum_{i=1}^{\frac{n-2}{2}}c_iTr_1^n(x^{1+2^i})+Tr_1^{n/2}
(x^{1+2^{\frac{n}{2}}}),
\end{align}
where $c_i\in GF(2)$ and  $Tr_1^{n/2}(x)$ is the trace function from $GF(2^{\frac{n}{2}})$ to $GF(2)$. They proved that $f(x)$ is a bent function if and only if
$$
gcd(c(x),x^n+1)=1,
$$
where $c(x)=\sum_{i=1}^{\frac{n-2}{2}}c_i(x^i+x^{n-i})+x^{n/2}$. For some special cases of $n$, Yu and Gong \cite{YG}
considered the concrete constructions of bent functions of the form (\ref{mbent}) and presented some enumeration results.

Hu and Feng \cite{HF} generalized results of Ma, Lee and Zhang \cite{MLZ} and studied the quadratic  Boolean functions of the form
\begin{align}\label{hfbent}
f(x)=\sum_{i=1}^{\frac{m-2}{2}}c_iTr_1^n(\beta x^{1+2^{ei}})+Tr_1^{n/2}(\beta x^{1+2^{\frac{n}{2}}}),
\end{align}
where $c_i\in GF(2)$, $n=em$, $m$ is even and $\beta\in GF(2^e)$. They obtained that $f(x)$ is bent if and only if
$$gcd(c(x),x^m+1)=1,
$$
where $c(x)=\sum_{i=1}^{\frac{m-2}{2}}c_i(x^i+x^{m-i})+x^{m/2}$. Further, they presented the enumerations of bent functions for some specified $m$.
Note that $\beta\in GF(2^e)$, then $(\beta^{2^{e-1}})^{1+2^{ei}}=\beta^{2^e}=\beta$. The function $f(x)$ of the form (\ref{hfbent}) satisfies that
$$
f(x)=\sum_{i=1}^{\frac{m-2}{2}}c_iTr_1^n((\beta^{2^{e-1}} x)^{1+2^{ei}})+Tr_1^{n/2}((\beta^{2^{e-1}} x)^{1+2^{\frac{n}{2}}}),
$$
where $c_i\in GF(2)$. From the transformation $x\longmapsto \beta^{2^{e-1}}x$, a bent function of the form (\ref{hfbent}) is changed into a bent function  of the form  (\ref{mbent}). Actually, (\ref{hfbent}) does not introduce new bent functions.

In \cite{TQX}, we considered the bentness of quadratic Boolean functions of the form
\begin{align}
f(x)=\sum_{i=1}^{\frac{n}{2}-1}Tr^n_1(c_ix^{1+2^i})+ Tr_1^{n/2}(c_{n/2}x^{1+2^{n/2}}),\label{mgb}
\end{align}
where $c_i\in GF(2^n)$ for $1\leq i\leq \frac{n}{2}-1$ and
$c_{n/2}\in GF(2^{n/2})$. The bentness of these functions can be characterized by a class of linearized permutation polynomials. Further, we generalized quadratic Boolean functions in \cite{HF,MLZ} and study Boolean functions of the form
\begin{align}
f(x)=\sum_{i=1}^{\frac{m}{2}-1}Tr^n_1(c_ix^{1+2^{ei}})+Tr_1^{n/2}(c_{m/2}
x^{1+2^{n/2}})~\label{mmgb}
\end{align}
where $n=em$, $m$ is even and $c_i\in GF(2^e)$. For application, we presented some new quadratic bent functions and gave methods of constructing new quadratic bent functions from known quadratic bent functions. Finally, we presented enumerations of bent functions of the form (\ref{mmgb}) for the  case $m=2^{v_0}p^r$ and $gcd(e,p-1)=1$, where $v_0>0$, $r>0$, $p$ is an odd prime satisfying $ord_p(2)=p-1$ or $ord_p(2)=(p-1)/2$ ($(p-1)/2$ is odd).

In this paper, we consider the construction of quadratic bent functions of the form (\ref{mmgb}) for the following two cases:

{\rm (i)} $m=2^{v}p^r$ and $gcd(e,p-1)=1$, where $v>0$, $r>0$ and $p$ is an odd prime satisfying $ord_p(2)=p-1$ or $ord_p(2)=(p-1)/2$ ($(p-1)/2$ is odd).

{\rm (ii)}$m=2^vpq$, where $v>0$, $p, q$ are two odd primes satisfying that $gcd(p-1,q-1)=2$, $ord_p(2)=p-1$, $ord_q(2)=q-1$, $2|\frac{(p-1)(q-1)}{4}$ and $gcd(e,(p-1)(q-1))=1$.

Conditions of $c_i$ for making such a quadratic bent function are given.
Further, we have the results of Yu and Gong \cite{YG} by setting
$e=1$ in Case {\rm (i)}. And we present the enumeration for quadratic bent functions in Case {\rm(ii)}.

The rest of the paper is organized as follows: Section 2 gives some results on quadratic Boolean functions. Section 3 gives the construction and enumeration of quadratic bent functions for two cases. Section 4 concludes for this paper.

\section{Preliminaries}

In this section, some notations are given first. Let $GF(2^n)$ be the finite field
with $2^n$ elements. Let $GF(2^n)^*$ be the multiplicative group of $GF(2^n)$.
Let $e|n$, the trace function $Tr^n_e(x)$ from $GF(2^n)$ to $GF(2^e)$ is defined by
$$
Tr^n_e(x)=x+x^{2^e}+\cdots+x^{2^{e(n/e-1)}},~~~~x\in GF(2^n).
$$
The trace function satisfies that

{\rm (1)} $Tr^n_e(x^{2^e})=Tr^n_e(x)$, where $x\in GF(2^n).$

{\rm (2)} $Tr^n_e(ax+by)=aTr_{e}^{n}(x)+bTr_{e}^n(y)$,
where $x,y\in GF(2^n)$ and $a,b\in GF(2^e).$

When $n$ is even, a quadratic Boolean function from $GF(2^n)$ to $GF(2)$ can be represented by
\begin{align}
f(x)=\sum_{i=0}^{\frac{n}{2}-1}Tr^n_1(c_ix^{1+2^{i}})+Tr_1^{n/2}(c_{n/2}
x^{1+2^{n/2}}),\label{eq}
\end{align}
where $c_i\in GF(2^n)$ for $0\leq i \leq \frac{n}{2} $ and $c_{n/2}\in GF(2^{\frac{n}{2}}).$

When $n$ is odd, $f(x)$ can be represented by
\begin{align}
f(x)=\sum_{i=0}^{\frac{n-1}{2}}Tr^n_1(c_ix^{1+2^{i}}),\label{oq}
\end{align}
where $c_i\in GF(2^n)$.

For a Boolean function $f(x)$ over $GF(2^n)$, the Hadamard transform is defined by
$$
\hat{f}(\lambda)=\sum_{x\in GF(2^n)} (-1)^{f(x)+Tr^n_2(\lambda x)}, \lambda\in GF(2^n).
$$
For a quadratic Boolean function $f(x)$ of the form (\ref{eq}) or (\ref{oq}), the distribution of the  Hadamard transform can be described by the bilinear form
\begin{align*}
Q_f(x,y)=f(x+y)+f(x)+f(y).
\end{align*}
For the bilinear form $Q_f$, define
\begin{align}
K_f=\{x\in GF(2^n):Q_f(x,y)=0,\forall y \in GF(2^n)\} \label{Kf}
\end{align}
and $k_f=dim_{GF(2)}(K_f)$. Then $2|(n-k_f)$. The distribution of the Hadamard
transform values of
$\hat{f}(\lambda)$  is given in the following theorem \cite{HK}.
\begin{theorem}\label{vd}
Let $f(x)$ be a quadratic Boolean function of the form (\ref{eq}) or (\ref{oq}) and $k_f=dim_{GF(2)}(K_f)$, where $K_f$ is defined in (\ref{Kf}). The distribution of the Hadamard transform values of $f(x)$ is given by
\begin{align*}
\hat{f}(\lambda)=\begin{cases}0,&2^{n}-2^{n-k_f}~times \cr 2^{\frac{n+k_f}{2}},& 2^{n-k_f-1}+2^{\frac{n-k_f}{2}-1} ~times\cr -2^{\frac{n+k_f}{2}},& 2^{n-k_f-1}-2^{\frac{n-k_f}{2}-1}~times.\end{cases}
\end{align*}
\end{theorem}
Bent functions as an important class of Boolean functions are defined below.
\begin{definition}
Let $f(x)$ be a Boolean function from $GF(2^n)$ to $GF(2)$. Then $f(x)$ is called a bent function if for any $\lambda \in GF(2^n)$, $\hat{f}(\lambda)\in \{2^{\frac{n}{2}},-2^{\frac{n}{2}}\}$.
\end{definition}
Bent functions only exist in the case for even $n$. From Theorem \ref{vd}, the following result on bent functions is given below.
\begin{corollary}\label{bl}
Let $f(x)$ be a quadratic function of the form (\ref{eq}) over $GF(2^n)$, then $f(x)$ is bent if and only if $K_f=\{0\}$, where $K_f$ is defined in (\ref{Kf}).
\end{corollary}

In \cite{TQX}, we presented the characterization of bentness of quadratic Boolean functions defined by (\ref{mmgb}).
\begin{theorem} \label{gcd}
Let $n=em$ and $m$ be even.
Let $f(x)$ be a Boolean function defined by
\begin{align*}
f(x)=\sum_{i=1}^{\frac{m}{2}-1}Tr^n_1(c_ix^{1+2^{ei}})+
Tr_1^{n/2}(c_{m/2}x^{1+2^{n/2}})
\end{align*}
where $c_i\in GF(2^e)$, then $f(x)$
is bent if and only if $gcd(c_f(x),x^m+1)=1$, where
\begin{align}\label{kp}
c_f(x)=\sum_{i=1}^{\frac{m}{2}-1}c_i(x+x^{m-i})+c_{m/2}x^{m/2}.
\end{align}
In particular, if $f(x)$ is bent, then $c_{m/2}\neq 0$.
\end{theorem}
For a special case, we have the following theorem.
\begin{theorem}
Let $m=2^{v_0}p^r$, where  $v_0\geq 1$,   $r\geq 1$ and $p$ is an odd prime satisfying that  that $ord_p(2)=p-1$ or $ord_p(2)=\frac{p-1}{2}$~$(\frac{p-1}{2}\text{~ is odd})$. Let $gcd(e,p-1)=1$.
Define the Boolean function
\begin{align*}
f(x)=\sum_{i=1}^{\frac{m}{2}-1}Tr^n_1(c_ix^{1+2^{ei}})
+Tr_1^{n/2}(c_{m/2}x^{1+2^{n/2}})
\end{align*}
where $c_i\in GF(2^e)$. The number of bent functions of this form  is
$$
(2^e-1)2^{e\frac{m-2}{2}}\prod_{i=1}^{r}(1-(\frac{1}{2^e})^{\frac{p^i-p^{i-1}}{2}}).
$$
\end{theorem}

\section{Constructions of quadratic bent functions}

In this section, we consider quadratic Boolean functions of the form
\begin{align}
f(x)=\sum_{i=1}^{\frac{m}{2}-1}Tr^n_1(c_ix^{1+2^{ei}})+
Tr_1^{n/2}(c_{m/2}x^{1+2^{n/2}}) ,~c_i\in GF(2^e)\label{OURC}
\end{align}
for two cases.
\subsection{The case: $m=2^vp^r$}
In this subsection, let $n=em$, $m=2^{v}p^r$ and $gcd(e,p-1)=1$, where $v>0$, $r>0$ and $p$ is an odd prime satisfying that  $ord_p(2)=p-1$ or $ord_p(2)=\frac{p-1}{2}~((p-1)/2\text{ is odd})$.
For the characterization of the bentness of quadratic Boolean functions, some results on cyclotomic polynomials are given first \cite{LN}. The
$d$-th cyclotomic polynomial $Q_d(x)$, whose roots are  primitive $d$-th roots  of unity, is a monic polynomial of order $d$ and degree $\phi(d)$($\phi(\cdot)$
is Euler-totient function).
\begin{lemma}\label{cycQ}
{\rm (1)} $Q_p(x)=\frac{x^p+1}{x+1}=1+x+x^2+\cdots x^{p-1}$.

{\rm(2)} $x^N+1=\prod_{d|N}Q_d(x)$. In particular,
\begin{align*}
x^{p^r}+1=&(x+1)Q_{p}(x)\cdots Q_{p^r}(x),\\
x^{pq}+1=&(x+1)Q_{p}(x)Q_q(x)Q_{pq}(x),
\end{align*}
where $p$ and $q$ are two different odd primes.

{\rm (3)} Let $p$ be a prime and  $gcd(e,p-1)=1$. Let $p$ satisfy  $ord_p(2)=p-1$ or $ord_p(2)=\frac{p-1}{2}$~$((p-1)/2\text{ is odd})$. Then
$gcd(c_f(x),Q_{p^i}(x))\neq 1$ if and only if
$gcd(c_f(x),Q_{p^i}(x))=Q_{p^i}(x)$. Equivalently, $gcd(c_f(x),Q_{p^i}(x))= 1$ if and only if $c_f(x)\not \equiv 0 \mod Q_{p^i}(x)$.

{\rm (4)} Let $p,q$ be two different primes, where $gcd(p-1,q-1)=2$, $ord_p(2)=p-1$, $ord_q(2)=q-1$, $2|\frac{(p-1)(q-1)}{4}$ and
$gcd(e,(p-1)(q-1))=1$, then $gcd(c_f(x),Q_{pq}(x))\neq 1$
if and only if $gcd(c_f(x),Q_{pq}(x))=Q_{pq}(x)$.
\end{lemma}
\begin{proof}
{\rm (1) and (2)} The proof can be found in \cite{B,LN}.

{\rm (3)} Since $gcd(e,p-1)=1$, then $ord_p(2^e)=ord_p(2)$. This can be obtained similarly from the proof of Lemma 4 in \cite{YG}.

{\rm (4)} From $ord_p(2)=p-1$ and $ord_q(2)=q-1$,  $ord_{pq}(2)=(p-1)(q-1)/2$. From $gcd(e,(p-1)(q-1))=1$,
$ord_{pq}(2^e)=(p-1)(q-1)/2$. Let $\alpha$ be a primitive $pq$-th root of unity, then $g(x)=\prod_{i=0}^{(p-1)(q-1)/2-1}(x+\alpha^{2^{ei}})$ and
$g^*(x)=\prod_{i=0}^{(p-1)(q-1)/2-1}(x+(\alpha^{-1})^{2^{ei}})$ are two irreducible polynomials over $GF(2^e)$. To explain that $g(x)$ and $g^*(x)$ are different polynomials, we just prove that $2^{ei}\not\equiv -1 \mod pq$ for any $0\leq i\leq (p-1)(q-1)/2-1$. If not, let $2^{ie}\equiv -1 \mod pq$
for some $i$, then $(2^e)^{2i}\equiv 1 \mod pq$ and
$\frac{(p-1)(q-1)}{2}|2i$. We have $i=0,\frac{(p-1)(q-1)}{4},2\cdot \frac{(p-1)(q-1)}{4},\cdots$ and $i=\frac{(p-1)(q-1)}{4}$. Since $2|\frac{(p-1)(q-1)}{4}$, either $\frac{p-1}{2}$ or $\frac{q-1}{2}$
is even. Suppose that $\frac{q-1}{2}$ is even. Then $2^{ei}\equiv (2^e)^{(p-1)(q-1)/4}$$\equiv (-1)^{(q-1)/2}$$\equiv 1\mod p$, which contradicts that $2^{ei}\equiv -1\mod p$. Hence, $2^{ei}\not\equiv -1 \mod pq$ for any $0\leq i\leq (p-1)(q-1)/2-1$.

Since $g(x)$ and $g^*(x)$ are different, then  $Q_{pq}(x)=g(x)g^*(x)$.  Further, $gcd(c_f(x),x^{pq}+1)\in \{1,g(x),g^*(x),Q_{pq}(x)\}$.
If $gcd(c_f(x),x^{pq}+1)\neq 1$, suppose that $g(x)|c_f(x)$, then $c_f(\alpha)=0$. Note that $c_f(\alpha^{-1})=\sum_{i=1}^{m/2-1}c_i((\alpha^{-1})^i+
(\alpha^{-1})^{-i})+c_{m/2}=c_f(\alpha)=0$, that is, $g^*(x)|c_f(x)$. Hence, $Q_{pq}(x)|c_f(x)$.
\end{proof}
\begin{lemma}\label{CPOLY}
Let $m=2^{v}p^r$ and $gcd(e,p-1)=1$, where $v>0$, $r>0$ and $p$ is an odd prime satisfying that $ord_p(2)=p-1$ or $ord_p(2)=\frac{p-1}{2}~((p-1)/2\text{ is odd})$. Let
$c_f(x)$ be defined in (\ref{kp}). Then
$gcd(c_f(x),x^m+1)=1$ if and only if the following conditions hold.

{\rm (i)} $c_f(1)\neq 0$.

{\rm (ii)} $(x^{p^{k-1}}+1)c_f(x)\not\equiv 0 \mod x^{p^k}+1$ for any $1\leq k\leq r$.
\end{lemma}
\begin{proof}
Since $x^m+1=(x^{p^r}+1)^{2^v}$, then $gcd(c_f(x),x^m+1)=1$ if and only if $gcd(c_f(x),x^{p^r}+1)=1$. From $x^{p^r}+1=(x+1)Q_p(x)\cdots Q_{p^r}(x)$, $gcd(c_f(x),x^{p^r}+1)=1$ if and only if $gcd(c_f(x),x+1)=gcd(c_f(x),Q_p(x))=\ldots=gcd(c_f(x),Q_{p^r}(x))=1$.
It is obviously verified that $gcd(c_f(x),x+1)=1$ if and only if $c_f(1)\neq 1$.

If $(x^{p^{k-1}}+1)c_f(x)\equiv 0 \mod x^{p^k}+1$ for some $k (1\leq k \leq r)$, there exists $A(x)\in GF(2^e)[x]$ satisfying that
$$
(x^{p^{k-1}}+1)c_f(x)=A(x)(x^{p^k}+1).
$$
Note that $\frac{x^{p^k}+1}{x^{p^{k-1}}+1}=Q_{p^{k}}(x)$, then
$$c_f(x)=A(x)Q_{p^k}(x),
$$
that is, $gcd(c_f(x),Q_{p^k}(x))=Q_{p^k}(x)\neq 1$.

If $gcd(c_f(x),Q_{p^k}(x))\neq 1$ for some $k (1\leq k \leq r)$. From Lemma \ref{cycQ}, $gcd(c_f(x),Q_{p^k}(x))=Q_{p^k}(x)$ and there exists $A(x)\in GF(2^e)[x]$ satisfying that
$$c_f(x)=A(x)Q_{p^k}(x).$$
Further,
$$(x^{p^{k-1}}+1)c_f(x)=A(x)(x^{p^k}+1),$$
that is,
$$(x^{p^{k-1}}+1)c_f(x)\equiv 0 \mod x^{p^k}+1.$$

Hence, this theorem follows.
\end{proof}
\begin{lemma}\label{comput}
Let $c_f(x)$ be defined in (\ref{kp}). Then

{\rm (1)} $c_f(1)=c_{m/2}$.

{\rm (2)} For any $1\leq k\leq r$,
\begin{align*}
c_{f,k}(x)\equiv &c_f(x)\\
\equiv & \sum_{i=1}^{p^k-1}w_{i,k} x^i+c_{m/2}\\
\equiv & \sum_{i=1}^{(p^k-1)/2}(w_{i,k}+w_{p^k-i,k}) (x^i+x^{p^k-i})+c_{m/2} \mod x^{p^k}+1,\\
\end{align*}
where
\begin{align*}
w_{i,k}=\sum^{\frac{m}{2p^k}-1}_{j=0}c_{i+jp^k},~(1\leq i\leq p^k-1).
\end{align*}

{\rm (3)} Let $i=i_0+jp^k$, where $0\leq i_0\leq p^k-1$. Define
\begin{align}
u_{i_0,k}=&\begin{cases}w_{i_0,k}+w_{p^k-i_0,k},~&1\leq i_0\leq p^k-1
\cr c_{m/2},~& i_0=0
\end{cases},\notag\\
u_{i,k}=&u_{i_0,k}.\label{uik}
\end{align}
Then
\begin{align*}
(x^{p^{k-1}}+1)c_f(x)\equiv &(x^{p^{k-1}}+1)c_{f,k}(x)\\
\equiv &\sum_{i=0}^{p^k-1} (u_{i,k}+u_{i-p^{k-1},k})x^i \mod x^{p^k}+1.
\end{align*}

{\rm (4)} Let $u_{i,k}$ be defined in (\ref{uik}). Let
\begin{align*}
U_k=\left[
  \begin{array}{cccccc}
    u_{0,k} & u_{1,k} & \cdots & u_{\frac{p^{k-1}-1}{2},k} & \cdots & u_{p^{k-1}-1,k} \\
    u_{p^{k-1},k} & u_{p^{k-1}+1,k} & \cdots & u_{p^{k-1}+\frac{p^{k-1}-1}{2},k} & \cdots & u_{2p^{k-1}-1,k} \\
    u_{2p^{k-1},k} & u_{2p^{k-1}+1,k} & \cdots & u_{2p^{k-1}+\frac{p^{k-1}-1}{2},k} & \cdots & u_{3p^{k-1}-1,k} \\
    \vdots & \vdots & \vdots & \vdots & \vdots & \vdots\\
    u_{(\frac{p-1}{2})p^{k-1},k} & u_{(\frac{p-1}{2})p^{k-1}+1,k} & \cdots & u_{\frac{p^{k}-1}{2},k}&\cdots&u_{\frac{p^{k}-1}{2}+\frac{p^{k-1}-1}{2},k} \\
    \vdots & \vdots & \vdots & \vdots & \vdots & \vdots\\
    u_{(p-1)p^{k-1},k} & u_{(p-1)p^{k-1}+1,k} & \cdots & u_{(p-1)p^{k-1}+\frac{p^{k-1}-1}{2},k} & \cdots & u_{p^{k}-1,k} \\
  \end{array}
\right]
\end{align*}
Let
\begin{align}\label{Aik}
U_k=\left[A_{0,k},A_{1,k},\cdots,A_{p^{k-1}-1,k}\right],
\end{align}
where
\begin{align*}
A_{i,k}=[A_{i,k}(0),A_{i,k}(1),\cdots,A_{i,k}(p-1)]'
\end{align*}
is the $i$-th column of $U_k$.
Then

{\rm (i)} $A_{i,k}(j)=A_{p^{k-1}-i,k}(p-1-j)$ for any $1\leq i\leq p^{k-1}-1$ and $0\leq j\leq p-1$.

{\rm(ii)} For any $1\leq i\leq p^{k-1}-1$, $A_{i,k}$ is a constant if and only if $A_{p^{k-1}-i,k}$ is a constant.
\end{lemma}
\begin{proof}
{\rm(1)} This can be obviously obtained from the definition of $c_f(x)$.

{\rm (2)} This can be obtained in Lemma 3 in \cite{YG}.

{\rm (3)} From the definition of $u_{i,k}$,
$$
c_f(x)\equiv \sum_{i=0}^{p^k-1} u_{i,k}x^i \mod x^{p^k}+1.
$$
Then
\begin{align*}
(x^{p^{k-1}}+1)c_{f}(x)\equiv &(x^{p^{k-1}}+1)\sum_{i=0}^{p^k-1} u_{i,k}x^i\\
\equiv &\sum_{i=0}^{p^k-1} u_{i,k}x^i+\sum_{i=0}^{p^k-1} u_{i,k}x^{i+p^{k-1}}  \mod x^{p^k}+1.
\end{align*}
Note that $u_{i+p^k}=u_i$. Then
\begin{align*}
(x^{p^{k-1}}+1)c_{f}(x)\equiv &\sum_{i=0}^{p^k-1} u_{i,k}x^i+\sum_{i=0}^{p^k-1} u_{i-p^{k-1},k}x^i\\
\equiv &\sum_{i=0}^{p^k-1} (u_{i,k}+u_{i-p^{k-1},k})x^i \mod x^{p^k}+1.
\end{align*}

{\rm(4)} Note that $u_{p^k-i}=u_i$, then
\begin{align*}
A_{p^{k-1}-i,k}(p-1-j)=&u_{p^{k-1}-i+(p-1-j)p^{k-1}}\\
=&u_{p^k-(i+jp^{k-1})}\\
=&u_{i+jp^{k-1}}\\
=&A_{i,k}(j).
\end{align*}
Result {\rm(i)} holds. From Result {\rm(i)}, Result {\rm(ii)} holds.
\end{proof}
\begin{theorem}\label{pr}
Let $m=2^{v}p^r$ and $gcd(e,p-1)=1$, where $v>0$, $r>0$ and $p$ is an odd prime satisfying that $ord_p(2)=p-1$ or $ord_p(2)=\frac{p-1}{2}~((p-1)/2\text{ is odd })$. The quadratic Boolean function defined in (\ref{OURC}) is bent if and only if both of the following conditions hold.

{\rm(i)} $c_{m/2}\neq 0$.

{\rm(ii)} For any $1\leq k \leq r$, there exists $0\leq i \leq \frac{p^{k-1}-1}{2}$  such that $A_{i,k}$ is not a constant, where $A_{i,k}$ is defined in (\ref{Aik}).
\end{theorem}
\begin{proof}
From {\rm(1)} in Lemma \ref{comput}, {\rm(i)} in this theorem is equivalent to  {\rm(i)} in Lemma \ref{CPOLY}. From Theorem \ref{gcd}, we just prove that {\rm(ii)} in this theorem is equivalent to {\rm(ii)} in Lemma \ref{CPOLY}.

Suppose that {\rm(ii)} in Lemma \ref{CPOLY} holds. If {\rm(ii)} in this theorem does not hold, then there exists  $1\leq k \leq r$ such that  for any $0\leq i \leq \frac{p^{k-1}-1}{2}$, $A_{i,k}$ is a constant vector. From {\rm(4)} in Lemma \ref{comput}, $A_{i,k}$ is a constant for any $0\leq i \leq p^{k-1}-1$. Equivalently, $u_{i,k}=u_{i-p^{k-1},k}$. From {\rm(3)} in Lemma \ref{comput}, $(x^{p^{k-1}}+1)c_f(x)\equiv 0 \mod x^{p^k}+1$, which contradicts the supposition.

Suppose {\rm(ii)} in this theorem holds. If {\rm(ii)} in Lemma \ref{CPOLY}
does not hold, there exists $1\leq k \leq r$ such that $(x^{p^{k-1}}+1)c_f(x)\equiv 0 \mod x^{p^k}+1$. From {\rm(3)} in Lemma
\ref{comput}, $A_{i,k}$ is a constant for any $i$ satisfying that $u_{i,k}+u_{i-p^{k-1},k}=0$. Thus, this makes a contradiction.

Hence, this theorem follows.
\end{proof}
\begin{corollary}\label{2vp}
Let $m=2^{v}p$ and $gcd(e,p-1)=1$, where $v>0$, $p$ is an odd prime satisfying that $ord_p(2)=p-1$ or $ord_p(2)=\frac{p-1}{2}~((p-1)/2\text{ is odd})$.
Then the quadratic Boolean function defined in (\ref{OURC}) is bent if and only if both of the following conditions hold.

{\rm(i)} $c_{m/2}\neq 0$.

{\rm(ii)} There exists $1 \leq i\leq (p-1)/2$ such that $w_i+w_{p-i}\neq c_{m/2}$, where $w_i=\sum_{j=0}^{2^{v-1}-1}c_{i+jp}$.
\end{corollary}
\begin{proof}
From Theorem \ref{pr}, when $r=1$, $U_1=[A_{0,1}]$. Note that
$$
A_{0,1}=[c_{m/2},w_1+w_{p-1},\cdots,w_{(p-1)/2}+w_{(p+1)/2},\cdots,w_{p-1}
+w_{1}]'.
$$
Then $A_{0,1}$ is not a constant if and only if there exists $1\leq i\leq (p-1)/2$  such that $w_i+w_{p-i}\neq c_{m/2}$. Hence, this lemma follows.
\end{proof}
\begin{corollary}
Let $m=2p$ and $gcd(e,p-1)=1$, where $p$ is an odd prime satisfying that $ord_p(2)=p-1$ or $ord_p(2)=\frac{p-1}{2}~((p-1)/2\text{ is odd})$. Then
the quadratic Boolean function defined in (\ref{OURC}) is bent if and only if both of the following conditions holds.

{\rm(i)} $c_{m/2}\neq 0$.

{\rm(ii)} there exists $1 \leq i\leq (p-1)/2$ such that $c_i+c_{p-i}\neq c_{m/2}$.
\end{corollary}
\begin{proof}
Note that $w_i=c_i$ when $m=2p$. From Lemma \ref{2vp}, this corollary follows.
\end{proof}

\subsection{The case: $m=2^vpq$}
\begin{lemma}\label{CPOLYpq}
Let $m=2^vpq$, where $v>0$, $p,q$ are two different odd primes satisfying that $gcd(p-1,q-1)=2$, $ord_p(2)=p-1$, $ord_q(2)=q-1$, $2|\frac{(p-1)(q-1)}{4}$ and $gcd(e,(p-1)(q-1))=1$. Let
$c_f(x)$ be defined in (\ref{kp}). Then
$gcd(c_f(x),x^m+1)=1$ if and only if the following conditions hold.

{\rm(i)} $c_f(1)\neq 0$;

{\rm(ii)} $(x+1)c_f(x)\not \equiv 0\mod x^p+1;$

{\rm(iii)} $(x+1)c_f(x)\not \equiv 0\mod x^q+1;$

{\rm(iv)} $(x^p+1)(x^q+1)c_f(x)\not \equiv 0\mod x^{pq}+1;$
\end{lemma}
\begin{proof}
From {\rm(2)} in Lemma \ref{cycQ}, $x^m+1=[(x+1)Q_p(x)Q_q(x)Q_{pq}(x)]^{2^v}$. Then  $gcd(c_f(x),x^m+1)=1$  if and only if the following conditions hold.

{\rm(i')} $gcd(c_f(x),x+1)=1$;

{\rm(ii')} $gcd(c_f(x),Q_p(x))=1$;

{\rm(iii')} $gcd(c_f(x),Q_q(x))=1$;

{\rm(iv')} $gcd(c_f(x),Q_{pq}(x))=1$.

Similar to the proof of Lemma \ref{CPOLY}, we have {\rm(i)} and {\rm(i')}, {\rm(ii)} and {\rm(ii')}, {\rm(iii)} and {\rm(iii')} are equivalent correspondingly. We now prove that {\rm(iv)} and {\rm(iv')} are equivalent.

Suppose that {\rm(iv)} holds. If {\rm(iv')} does not hold, from {\rm(4)}  in Lemma \ref{cycQ}, there exists $A(x)\in GF(2^e)[x]$ such that
$c_f(x)=A(x)Q_{pq}(x)$. Then
$$
(x^p+1)(x^q+1)c_f(x)=A(x)(x^p+1)(x^q+1)Q_{pq}(x)=(x+1)A(x)(x^{pq}+1),
$$
which contradicts {\rm(iv)}.

Suppose that {\rm(iv')} holds. If {\rm(iv)} does not hold, there exists $A(x)\in GF(2^e)[x]$ such that
$$(x^p+1)(x^q+1)c_f(x)=A(x)(x^{pq}+1).$$
Note that $(x+1)|A(x)$, then
$$c_f(x)=\frac{A(x)}{x+1}Q_{pq}(x),$$
which contradicts {\rm(iv')}.

Hence, this lemma follows.
\end{proof}

For further discussion, let
\begin{align*}
w_{i}^{p}=&\sum_{j=0}^{\frac{m}{2p}-1}c_{i+jp},~&1\leq i\leq p-1\\
w_{i}^{q}=&\sum_{j=0}^{\frac{m}{2q}-1}c_{i+jq},~&1\leq i\leq q-1\\
w_{i}^{pq}=&\sum_{j=0}^{\frac{m}{2pq}-1}c_{i+jpq},~&1\leq i\leq pq-1
\end{align*}
Let
\begin{align*}
u_{i}^{p}=&\begin{cases}w_{i}^{p}+w_{p-i}^p,~&1\leq i\leq p-1 \cr c_{m/2},~&i=0\end{cases}\notag\\
u_{i+jp}^{p}=&u_{i}^{p}\\
u_{i}^{q}=&\begin{cases}w_{i}^{q}+w_{q-i}^q,~&1\leq i\leq q-1 \cr c_{m/2},~&i=0\end{cases}\notag\\
u_{i+jq}^{q}=&u_{i}^{q}\\
u_{i}^{pq}=&\begin{cases}w_{i}^{pq}+w_{pq-i}^{pq},~&1\leq i\leq pq-1 \cr c_{m/2},~&i=0\end{cases}\notag\\
u_{i+jpq}^{pq}=&u_{i}^{pq}
\end{align*}
\begin{lemma}
\label{computpq}
Let $c_f(x)$ be defined by (\ref{kp}). Then
\begin{align*}
(x+1)c_f(x)\equiv &\sum_{i=0}^{p-1}(u_{i}^{p}+u_{i-1}^{p})x^i,\mod x^p+1\\
(x+1)c_f(x)\equiv &\sum_{i=0}^{q-1}(u_{i}^{q}+u_{i-1}^{q})x^i,\mod x^q+1\\
(x^p+1)(x^q+1)c_f(x)\equiv &\sum_{i=0}^{pq-1}(u_{i}^{pq}+u_{i-p}^{pq}+u_{i-q}^{pq}+u_{i-p-q}^{pq})x^i, \mod x^{pq}+1.
\end{align*}
\end{lemma}
\begin{proof}
From the similar discussion of {\rm(2)} and {\rm(3)} in Lemma \ref{comput}, this lemma follows.
\end{proof}
\begin{theorem}
Let $m=2^vpq$, where $v>0$, $p,q$ are two different odd primes satisfying that $gcd(p-1,q-1)=2$, $ord_p(2)=p-1$, $ord_q(2)=q-1$, $2|\frac{(p-1)(q-1)}{4}$, and $gcd(e,(p-1)(q-1))=1$. Then
the quadratic Boolean function defined in (\ref{OURC}) is bent if and only if the following four conditions hold.

{\rm(i)} $c_{m/2}\neq 0;$

{\rm(ii)} There exists $1\leq i\leq (p-1)/2$ such that $w_{i}^p+w_{p-i}^p\neq c_{m/2}$;

{\rm(iii)} There exists $1\leq i\leq (q-1)/2$ such that $w_{i}^q+w_{q-i}^q\neq c_{m/2}$;

{\rm(iv)} There exists $0\leq i\leq pq-1$ such that $u_{i}^{pq}+u_{i-p}^{pq}+u_{i-q}^{pq}+u_{i-p-q}^{pq}\neq 0$.
\end{theorem}
\begin{proof}
From {\rm(1)} in Lemma \ref{comput} and Lemma \ref{computpq}, the four conditions in Lemma \ref{CPOLYpq} are equivalent to four conditions in this theorem correspondingly. From Theorem \ref{gcd}, this theorem follows.
\end{proof}
\begin{theorem}
Let $m=2^vpq$, where $v>0$, $p,q$ are two different odd primes satisfying that  $gcd(p-1,q-1)=2$, $ord_p(2)=p-1$, $ord_q(2)=q-1$, $2|\frac{(p-1)(q-1)}{4}$ and
$gcd(e,(p-1)(q-1))=1$. Let $f(x)$ be a quadratic Boolean function defined by
\begin{align}\label{cmbent}
f(x)=\sum_{i=1}^{\frac{m}{2}-1}Tr^n_1(c_ix^{1+2^{ei}})+Tr_1^{n/2}
(c_{m/2}x^{2^{1+n/2}})
\end{align}
where $c_i\in GF(2^e)$. The number of quadratic bent functions of the form
(\ref{cmbent}) is
$$
(2^e-1)2^{e\frac{m-2}{2}}(1-(\frac{1}{2^e})^{\frac{p-1}{2}})(1-
(\frac{1}{2^e})^{\frac{q-1}{2}})(1-(\frac{1}{2^e})^{\frac{(p-1)(q-1)}{2}}).
$$
\end{theorem}
\begin{proof}
Similar to the proof of Theorem 4.6 in \cite{TQX}, we can have that the number of quadratic bent functions of the form (\ref{cmbent}) is
$$
\#(GF(2^e)^*)\#(\mathfrak{P}_{m-2}(\frac{x^{pq}+1}{x+1})),
$$
where $\mathfrak{P}_{*}(*)$ is defined in \cite{TQX}. From Lemma \ref{cycQ},
$$
\frac{x^{pq}+1}{x+1}=Q_p(x)Q_q(x)Q_{pq}(x)
$$
satisfies the factorization of  (21) in \cite{TQX}. From Lemma 4.4 in
\cite{TQX},
$$
\#(\mathfrak{P}_{m-2}(\frac{x^{pq}+1}{x+1}))=2^{e\frac{m-2}{2}}(1-
(\frac{1}{2^e})^{\frac{p-1}{2}})(1-(\frac{1}{2^e})^{\frac{q-1}{2}})
(1-(\frac{1}{2^e})^{\frac{(p-1)(q-1)}{2}}).$$
Hence, this theorem follows.
\end{proof}
\section{Conclusion}

In this paper, we generalize the work of Yu and Gong \cite{YG} and present conditions for the bentness of quadratic Boolean functions of the form (\ref{OURC}) for two cases:

{\rm(1)} $m=2^{v}p^r$ and $gcd(e,p-1)=1$, where $v>0$, $r>0$ and $p$ is an odd prime satisfying that $ord_p(2)=p-1$ or $ord_p(2)=\frac{p-1}{2}~((p-1)/2\text{ is odd})$;

{\rm(2)} $m=2^vpq$, where $v>0$, $p,q$ are two different odd primes satisfying that  $gcd(p-1,q-1)=2$, $ord_p(2)=p-1$, $ord_q(2)=q-1$, $2|\frac{(p-1)(q-1)}{4}$ and $gcd(e,(p-1)(q-1))=1$.

For Case {\rm(1)}, the numeration of quadratic bent functions is presented in   \cite{TQX}. For Case {\rm(2)}, we present the enumeration of quadratic bent functions in this paper.

\section*{Acknowledgment}

The authors
acknowledge support from
the Natural Science Foundation of China
(Grant No.10990011 \& No. 61272499). Yanfeng Qi also acknowledges support from Aisino Corporation Inc.





\bibliographystyle{model1c-num-names}







\end{document}